\documentclass[smallextended]{svjour3} 

\usepackage{graphics} 
\usepackage{epsfig} 
\usepackage{mathptmx} 
\usepackage{times} 
\usepackage{amsmath} 
\usepackage{amssymb}  
\usepackage{cite}
\usepackage{multirow}
\usepackage{mathtools}

\usepackage{amsthm}
\usepackage{slashbox}
\newtheorem{algorithm}[theorem]{Algorithm}

\title{\LARGE \bf
Security Strategies of Both Players in Asymmetric Information Zero-Sum Stochastic Games with an Informed Controller}

\author{Lichun Li, Cedric Langbort and Jeff S. Shamma \thanks{This work was supported in part by NSF grants 1619339 and 1151076 to Cedric Langbort, and by funding from King Abdullah University of Science and Technology (KAUST).
}}
\institute{Lichun Li and Cedric Langbort are with the Coordinated Science Lab, University of Illinois at Urbana-Champaign. {\tt\small lichunli,langbort@illinois.edu}. Jeff S. Shamma is with King Abdullah University of Science and Technology (KAUST),{\tt\small jeff.shamma@kaust.edu.sa}.}

\begin{document}

\maketitle
\thispagestyle{empty}

\begin{abstract}
This paper considers a zero-sum two-player asymmetric information stochastic game where only one player knows the system state, and the transition law is controlled by the informed player only. For the informed player, it has been shown that the security strategy only depends on the belief and the current stage. We provide LP formulations whose size is only linear in the size of the uninformed player's action set to compute both history based and belief based security strategies. For the uninformed player, we focus on the regret, the difference between $\mathbf{0}$ and the future payoff guaranteed by the uninformed player in every possible state. Regret is a real vector of the same size as the belief, and depends only on the action of the informed player and the strategy of the uninformed player. This paper shows that the uninformed player has a security strategy that only depends on the regret and the current stage. LP formulations are then given to compute the history based security strategy, the regret at every stage, and the regret based security strategy. The size of the LP formulations are again linear in the size of the uninformed player action set. Finally, an intrusion detection problem is studied to demonstrate the main results in this paper.
\end{abstract}

\section{INTRODUCTION}
Cyber attacks have been a serious threat to the security and privacy of individuals (e.g. Equifax data breach), companies (e.g. HBO cyberattack and Sony Pictures hack), and nations (e.g. stuxnet), and are reported to spur billions of dollars in loss \cite{cyberattackloss}. Such cyber attacks have become more stealthy, targeted, and sophisticated over the past few years. One difficulty in modelling and defending against them is that attackers often have access to a vast amount of attacking measures, which results in lack of complete information on the defender's part. Hence, we propose to model cyber security problems as games with asymmetric information, and get a systematic strategy to fight against cyber attacks.


A key element in asymmetric information games is to estimate the private information of the other players. This element is usually a probability, which is also called belief, over the other player's private information based on the history of observations. Generally speaking, a belief over the other player's private information depends on the player's strategy, which, in turn, depends on the belief. Therefore, there is always a coupling between the belief and the strategy. To decompose the coupling, common information based belief and the corresponding strategy were proposed \cite{nayyar2014common, ouyang2017dynamic, sinha2016structured, nayyar2017information}. In \cite{nayyar2014common}, Bayesian Nash equilibrium was considered. To decouple the belief from the strategy, it was assumed that the belief was strategy independent. With this assumption, asymmetric information games can be transformed to a symmetric game in which a backward induction was derived, and the Bayesian Nash equilibrium can be found by solving a one-stage Bayesian game. The idea was adopted in \cite{nayyar2017information} with a focus on zero-sum stochastic games. Both \cite{ouyang2017dynamic} and \cite{sinha2016structured} considered perfect Bayesian equilibrium which consists of a belief system and a strategy profile. The belief and the strategy need to be consistent with each other to form a perfect Bayesian equilibrium. In \cite{ouyang2017dynamic,sinha2016structured}, players' strategies were assumed to be known by each other. Based on this assumption, Ouyang et.al decomposed a stochastic game with asymmetric information and used a backward induction to find common information based perfect Bayesian equilibrium  \cite{ouyang2017dynamic}. Sinha and Anastasopoulos studied an infinite horizon discounted asymmetric information game in \cite{sinha2016structured}, and found that the common information based belief and strategy are stationary. A methodology was developed to decompose the interdependence between the belief and strategy, and to evaluate structured perfect Bayesian equilibrium.

While many previous work focused on beliefs in asymmetric information games, there is another group of works pointing out another key element in asymmetric information games \cite{de1996repeated,rosenberg1998duality,sorin2002first}. This element is a real vector of the same size as the belief, and does not depend explicitly on the other player's strategy. We call this vector `regret', because it is the difference between $\mathbf{0}$ and the future payoff guaranteed by a security strategy for every possible initial private information of the other players \cite{li2017computing,li2017efficient,li2017solving}. It was shown that the player without private information (uninformed player) has a security strategy that only depends on the regret in repeated games, plus the current stage if this is a finite stage game \cite{de1996repeated}. 

This paper focuses on asymmetric information zero-sum two-player stochastic games where only one player (informed) has access to private information (system state) which evolves following a Markovian rule controlled by the informed player only. Our goal is to provide tractable conditions for the computation of \textit{both players'} security strategies (to be defined precisely in Section \ref{section: problem statement}) for such games. More precisely, we show how to obtain LP formulations whose size is only linear in the cardinality of the uninformed player's action set, in contrast with existing approaches which do not consider the uninformed player's strategy and/or require LPs with size scaling polynomially in the cardinality of that set \cite{koller1996efficient}. For the informed player, our approach builds on the work of \cite{renault2006value}, which showed that the informed player has a security strategy that only depends on the belief and the current stage, and is independent of the action history of the uninformed player. We extend our original contribution \cite{li2014lp} by introducing an algorithm to compute belief based security strategy for the informed player. For the uninformed player, we introduce and build on the new notion of `regret', which generalizes the similar object we first considered in the context of repeated games \cite{li2017computing}. By using the dual game of the asymmetric information stochastic game, we show that in finite horizon asymmetric information stochastic games, the uninformed player has a security strategy that only depends on the regret and the current stage, and is independent of the history action of the uninformed player. The regret only depends on the action history of the informed player, and is independent of the strategy of the informed player. It is because of this property that an appropriately-sized LP can be derived, along with algorithms to compute the regret at every stage.

This paper is organized as follows. Section \ref{section: problem statement} introduces the game model. Section \ref{section: informed} and \ref{section: uninformed} introduce the security strategies of informed and uninformed players, respectively, and detail the derivations mentioned above. Finally, in Section \ref{section: simulation}, we apply our security strategy computation techniques to a game model of an intrusion detection problem.

\section{Problem statement}
\label{section: problem statement}
Let $\mathbb{R}^n$ denote the $n$-dimensional real space. For a finite set $K$, $| K |$ denotes its cardinality, and $\Delta(K)$ indicates the set of probability distributions over $K$. The symbols $\mathbf{1}$ and $\mathbf{0}$ denote vectors with all elements equal to 1 and 0, respectively. The size will be implied from context. For a vector $p$ and a matrix $Z$, we use $p(i)$ to denote the $i$th element of $p$, and $Z(i,j)$ to denote the element at the $i$th row and $j$th column of $Z$. The $i$th row and the $j$th column of $Z$ are denoted as $Z(i,:)$ and $Z(:,j)$, respectively.

A \emph{two-player zero-sum stochastic game} is specified by a six-tuple $(K,A,B,M,p_0$ $,Q)$, where
\begin{itemize}
  \item $K$ is a finite set, called the state set, whose elements are the states of the game.
  \item $A$ and $B$ are the finite action sets of player 1 and player 2, respectively.
  \item $M_k\in \mathbb{R}^{|A|\times|B|}$ is the payoff matrix if the state is $k\in K$. $M_k(a,b)$ is player 1's one stage payoff, or player 2's one stage cost if the current state is $k$ and the current actions of player 1 and 2 are $a$ and $b$, respectively.
  \item $p_0 \in \Delta(K)$ is the initial probability of the state.
  \item $Q_a \in \mathbb{R}^{|K|\times |K|}$ denotes the transition matrix if player 1 plays $a\in A$. $Q_a(k,k')$ is the conditional probability that the next state is $k'$ given the current action is $a$ and the current state is $k$.
\end{itemize}

An $N$-stage asymmetric information stochastic game with a single controller is played as follows. At the beginning of stage $t= 1,\ldots,N$, the state $k_t$ is chosen by nature according to the initial probability $p_0$ if this is the first stage, or the transition law $Q_{a_{t-1}}(k_{t-1},:)$ otherwise. The current state $k_t$ is only observed by player 1, and hence player 1 is called the informed player while player 2 is called the uninformed player. Both players choose their actions $a_t$ and $b_t$ simultaneously, which are observable by both players. The resulting one stage payoff of player 1, i.e. the one stage cost of player 2, is $M_{k_t}(a_t,b_t)$. We assume both players have perfect recall.

At the beginning of stage $t$, the available state history and action history of players 1 and 2 are indicated by $S_t=\{k_1,\ldots,k_t\}$, $I_t=\{a_1,\ldots,a_{t-1}\}$ and $J_t=\{b_1,\ldots,b_{t-1}\}$, respectively. Player 1's behavior strategy is an element $\sigma=(\sigma_t)_{t=1}^N$, where for each $t$, $\sigma_t: K^t\times A^{t-1}\times B^{t-1} \rightarrow \Delta(A)$. Player 2's behavior strategy is an element $\tau=(\tau_t)_{t=1}^N$, where for each $t$, $\tau_t: A^{t-1}\times B^{t-1} \rightarrow \Delta(B)$. Denote by $\Sigma_N$ and $\mathcal{T}_N$ the set of $N$-stage strategies of player $1$ and $2$, respectively.

Every quadruple $(p_0,\sigma,\tau,Q)$ induces a probability $P_{p_0,\sigma,\tau,Q}$ over the set of plays $(K \times A\times B)^N$. We denote by $\mathbf{E}_{p_0,\sigma,\tau,Q}$ the corresponding expectation operator. The \emph{total payoff} of the $N$-stage asymmetric information stochastic game is defined as
\begin{align}
  \gamma_N(p_0,\sigma,\tau)=\mathbf{E}_{p_0,\sigma,\tau,Q}\left(\sum_{t=1}^N M_{k_t}(a_t,b_t) \right)
\end{align}
The $N$-stage asymmetric information stochastic game $\Gamma_N(p_0)$ is defined as the zero-sum game with strategy spaces $\Sigma_N$ and $\mathcal{T}_N$, and payoff function $\gamma_N(p_0,\sigma,\tau)$.

In this game, player 1 wants to maximize the total payoff, while player 2 wants to minimize it. Therefore, player 1 has a security level $\underline{v}_N(p_0)$, which is also called the maxmin value of the game and defined as
\begin{align*}
  \underline{v}_N(p_0)=\max_{\sigma\in \Sigma_N}\min_{\tau\in \mathcal{T}_N} \gamma_N(p_0,\sigma,\tau).
\end{align*}
A strategy $\sigma^*$ that guarantees player 1's security level, i.e. $\min_{\tau\in \mathcal{T}} \gamma_N(p_0,\sigma^*,\tau)=\underline{v}_N(p_0)$, is called a security strategy of player 1. Player 2 also has a security level $\bar{v}_n(p_0)$ which is defined as
\begin{align*}
  \bar{v}_N(p_0)=\min_{\tau\in \mathcal{T}_N}\max_{\sigma\in \Sigma_N}\gamma_N(p_0,\sigma,\tau).
\end{align*}
Player 2's security level is also called the minmax value of the game, and a strategy $\tau^*\in \mathcal{T}$ that guarantees the security level of player 2 is a security strategy of player 2. Since this is a finite game (finite horizon, action sets, and state set) and behavior strategies are considered, its maxmin value and minmax value match \cite{sorin2002first}. In this case, we say the game has a value $v_N(p_0)=\underline{v}_N(p_0)=\bar{v}_N(p_0)$, and the security strategy pair $(\sigma^*,\tau^*)$ is the saddle point of the game.

\section{Security strategies of the informed player}
\label{section: informed}
The security strategies of the informed players in asymmetric information stochastic games have been thoroughly studied in previous work \cite{rosenberg2004stochastic,renault2006value,li2014lp}. For completeness of this paper, we will review the related results in this section. Interested readers can find proofs in the corresponding references.

\subsection{History based security strategy and its LP formulation}
\label{subsec: history, informed}

Renault showed that the informed player's security strategy only depends on the current state and its own action history \cite{renault2006value}. We state this property in the following lemma.
\begin{proposition}[Proposition 5.1 in \cite{renault2006value}]
  \label{lemma: H^B independent strategy, informed}
Consider an $N$-stage asymmetric information stochastic game $\Gamma_N(p_0)$. The informed player has a security strategy that, at every stage $t$, only depends on the current state $k_t$, and on the actions history $I_t$ of the informed player.
\end{proposition}

Based on this property, by mathematical induction, \cite{li2014lp} presented an LP whose size is only linear with respect to the size of the uninformed player's action set. Proposition \ref{lemma: H^B independent strategy, informed} indicates that there is no loss of generality in only considering the informed player's behavior strategies that depend on $k_t$ and $I_t$ only. Therefore, for the rest of this paper, we only consider informed player's behavior strategy $\sigma_t$ as a function from $K\times A^{t-1}$ to $\Delta(A)$.

Before presenting the simplified LP, we first define a matrix variable $Z_{I_t}\in \mathbb{R}^{|A|\times |K|}$ and a scalar variable $\ell_{I_t}\in \mathbb{R}$. Let $Z=(Z_{I_t})_{I_t\in A^{t-1}, t=1,\ldots,N}$ and $\ell=(\ell_{I_t})_{I_t\in A^{t-1}, t=1,\ldots,N}$. Denote the sets of all possible values that $Z$ and $\ell$ can take by $\mathcal{Z}$ and $L$. The history based security strategy of the informed player can be computed according to the following theorem.
\begin{theorem}[Theorem III.3 in \cite{li2014lp}]
  \label{theorem: history based strategy}
Consider an $N$-stage asymmetric information stochastic game $\Gamma_N(p_0)$ with the initial probability $p_0$. The game value $v_N(p_0)$ of $\Gamma_N(p_0)$ satisfies
\begin{align}
  v_N(p_0)=&\max_{Z\in \mathcal{Z},\ell\in L}\sum_{t=1}^N \sum_{I_t\in A^{t-1}} \ell_{I_t} \label{eq: LP informed 1}\\
  s.t.& \sum_{k\in K}M_k^T Z_{I_t}(:,k) \geq  \ell_{I_t} \mathbf{1}, && \forall I_t\in A^{t-1} ,\forall t=1,\ldots,N \\
&  \mathbf{1}^T Z_{I_t}(:,k) = Z_{I_{t-1}}(a,:)Q_a(:,k),&& \forall I_t=(I_{t-1},a)\in A^{t-1},\forall k\in K,\\
&&&\forall t=2,\ldots,N \\
& \mathbf{1}^T Z_{I_1}(:,k)=p_0(k),&& \forall k\in K \\
& Z_{I_t}(:,k) \geq  \mathbf{0} , &&\forall k\in K, \forall I_t\in A^{t-1},\forall t=1,\ldots,N \label{eq: LP informed 5}
\end{align}

Moreover, a security strategy $\sigma^*_t(k,I_t)$ of the informed player at stage $t$ is
\begin{align}
  \sigma^*_t(k,I_t)=\left\{
                        \begin{array}{ll}
                       \frac{Z^*_{I_t}(:,k)}{ \mathbf{1}^T Z_{I_t}^*(:,k)}, & \hbox{if $Z_{I_t}^*(:,k)\neq \mathbf{0}$;} \\
                       \mathbf{0}, & \hbox{$\mathrm{othewise}$}
                        \end{array}
                       \right. \label{eq: security strategy informed}
\end{align}
where $Z^*$ is the optimal solution of the LP formulation (\ref{eq: LP informed 1}-\ref{eq: LP informed 5}).
\end{theorem}

\subsection{Belief based security strategy and its LP based algorithm}
The memory required to record the history based security strategy increases exponentially with $N$ in game $\Gamma_N(p_0)$. Therefore, a sufficient statistics based security strategy is of interest, especially when $N$ is large. When studying the game value of a finite stage game, Renault showed that at stage $t$, the sufficient statistics of the informed player is the stage index and the conditional probability $p_t$ of the current state given the action history of the informed player. The conditional probability $p_t$ is also called the belief state which is updated as follows.
\begin{align}
  p_{t+1}=\phi^T(p_t,X_t,a)Q^{a}, \label{eq: belief update}
\end{align}
where $\phi:K\times \Delta(A)^{K} \times A \rightarrow \Delta(K)$ is a vector valued function whose $k$th element is
\begin{align}
  \phi_k(p_t,X_t,a)=\frac{p_t(k) X_t(a,k)}{\bar{x}(p_t,x_t,a)}, \forall k\in K \label{eq: p+}
\end{align}
$X_t(:,k)=\sigma_t(k,I_t)$, and $\bar{x}(p_t,X_t,a)=\sum_{k\in K}p_t(k) X_t(a,k)$ is the probability that player 1 plays $a$ at stage $t$.

Based on the belief state $p_t$, a recursive formula to compute the game value $v_N(p_0)$ was provided in \cite{renault2006value}, and the sufficient statistics of the informed player was also given at the same time.
\begin{theorem}[Proposition 5.1 and Remark 5.2 in \cite{renault2006value}]
  \label{theorem: sufficien statistics, informed}
Consider an $n$ stage asymmetric information stochastic game $\Gamma_n(p)$. Its game value $v_n(p)$ satisfies the following recursive formula.
\begin{align}
&v_n(p)\\
= &\max_{X\in \Delta(A)^{|K|}} \min_{\hat{y}\in\Delta(B)}\left(\sum_{k\in K}p(k) X^T(:,k) M_k \hat{y}+\sum_{a\in A}\bar{x}(p,x,a)v_{n-1}(\phi^T(p,X,a)Q_a)\right) \label{eq: recursive formula primal 1}\\
= &\min_{\hat{y}\in\Delta(B)} \max_{X\in \Delta(A)^{|K|}} \left(\sum_{k\in K}p(k) X^T(:,k) M_k \hat{y}+\sum_{a\in A}\bar{x}(p,x,a)v_{n-1}(\phi^T(p,X,a)Q_a) \right) \label{eq: recursive formula primal 2}
\end{align}
Moreover, the informed player has a security strategy at stage $t$ that only depends on stage $t$ and belief state $p_t$.
\end{theorem}

Based on this theorem, one can derive an algorithm to compute the belief based security strategy of the informed player as follows.
\begin{algorithm} \hfill{}
\label{Algorithm: belief based, informed}
\begin{enumerate}
  \item Initialization
    \begin{enumerate}
      \item Read payoff matrices $M$, transition matrices $Q$, time horizon $N$ and initial probability $p_0$.
      \item Set $t=1$ and $p_t=p_0$. Read $k_t$.
    \end{enumerate}
  \item Solve LP (\ref{eq: LP informed 1}-\ref{eq: LP informed 5}) by replacing $N$ and $p_0$ by $N+1-t$ and $p_t$. A security strategy at $t$ is $\sigma_1^*(k_t,I_1)$ computed according to (\ref{eq: security strategy informed}).
  \item Draw an action $a_t$ according to the security strategy $\sigma_1^*(k_t,I_1)$.
  \item Update $p_{t+1}$ according to (\ref{eq: belief update}).
  \item Update $t=t+1$, read $k_t$.
  \item If $t\leq N$, go to step 2). Otherwise, end.
\end{enumerate}
\end{algorithm}

Compared with history based security strategy, the belief based security strategy only needs to record stage $t$ and belief state $p_t$ whose size is fixed and much smaller than $I_t$ which is recorded in history based security strategy, especially when the time horizon $N$ is large. The belief based security strategy also provides us a research direction in dealing with infinite horizon asymmetric information stochastic games, which was studied in \cite{li2015efficient} for the discounted case.

\section{Security strategies of the uninformed player}
\label{section: uninformed}
While the security strategies of informed players in asymmetric information stochastic games were well studied in the previous work, only a few papers studied the security strategies of uninformed players \cite{de1996repeated,rosenberg1998duality}. Both references studied the security strategies of the uninformed player using the dual games of the corresponding game model. We follow a similar path in this section, noting that our game model is more general than that considered in \cite{de1996repeated}, and incomparable with that of \cite{rosenberg1998duality}.

\subsection{Regret based security strategies and $J_t$ independent security strategies}
\label{subsection: strategy property, uninformed}
Let us first introduce the dual game of the asymmetric information stochastic game $\Gamma_N(p)$. A dual asymmetric information stochastic game is specified by a six-tuple $(K,A,B,M,Q,\alpha)$, where $K,A,B,M,Q$ are defined in the same way as in the primal game $\Gamma_N(p)$, and $\alpha\in \mathbb{R}^{|K|}$ is the initial vector payoff of player 1. The dual game is played exactly in the same way as the primal game except that at the first stage, player 1, instead of Nature, chooses the state. Let $p\in \Delta(K)$ be player 1's mixed strategy to choose the initial state. The total payoff in the dual game is
\begin{align}
  g_N(\alpha,\sigma,\tau)=\mathbf{E}_{p,\sigma,\tau,Q}\left(\alpha(k_1)+\sum_{t=1}^N M_{k_t}(a_t,b_t)\right).
\end{align}
The $N$-stage dual asymmetric information stochastic game $G_N(\alpha)$ is defined as a two-player zero-sum game with strategy spaces $\Delta(K)\times \Sigma_N$ and $\mathcal{T}_N$, and payoff function $g_N(p,\sigma,\tau)$.

The dual game $G_N(\alpha)$ is still a finite game. Since behavior strategies are considered, the dual game has a value, i.e. $ w_N(\alpha)=\max_{\sigma\in \Sigma_N}\min_{\tau\in \mathcal{T}_N} g_N(\alpha,\sigma,\tau)$ $= \min_{\tau\in \mathcal{T}_N} \max_{\sigma\in \Sigma_N} g_N(\alpha,\sigma,\tau)$. Before studying the relation between the game values of the primal game and the dual game, we introduce the initial regret of the primal game $\Gamma_n(p)$ as follows. Denote the informed player strategy from stage $2$ to $N$ as $\sigma_{2:N}\in \Sigma_{2:N}$, where $\Sigma_{2:N}$ is the set of all possible values that $\sigma_{2:N}$ can take. Let $\tau^*$ be the uninformed player's security strategy in primal game $\Gamma_n(p)$. The initial regret $\hat{\alpha}_0\in \mathbb{R}^{|K|}$ of a primal game $\Gamma_n(p)$ is defined as
\begin{align}
\hat{\alpha}_0(k)=-\max_{\substack{\sigma_1(k,\emptyset)\in \Delta(A)\\ \sigma_{2:N}\in \Sigma_{2:N}}} \mathbf{E}_{(\sigma_1(k,\emptyset),\sigma_{2:N}),\tau^*,Q}(\sum_{t=1}^n M_{k_t}(a_t,b_t)|k_1=k).\label{eq: initial regret}
\end{align}
The $k$th element of the initial regret is the difference between $0$, the total payoff realized at the beginning of stage 1, and the security level that the uninformed player's security strategy can guarantee if the game state is $k$. Later, we will see that if we use the initial regret of primal game $\Gamma_n(p)$ as the initial vector payoff in the dual game, the security strategy of the uninformed player in the dual game is also the security strategy of the uninformed player in the primal game. Moreover, the game value $v_n(p)$ of the primal game equals to the game value $w_n(\hat{\alpha}_0)$ minus $p^T \hat{\alpha}_0$. In this way, we can evaluate the game value of the primal game from the game value of the dual game. Now, let us introduce the relations between the game values of the primal game and the dual game.
\begin{theorem}
  \label{theroem: game value relations}
Consider an $n$-stage asymmetric information stochastic game $\Gamma_n(p)$ and its dual game $G_n(\alpha)$. Let $v_n(p)$ and $w_n(\alpha)$ be the game values of $\Gamma_n(p)$ and $G_n(\alpha)$. We have
\begin{align}
  v_n(p)=\min_{\alpha\in \mathbb{R}^{|K|}}\{w_n(\alpha)-p^T \alpha\}, \label{eq: v n} \\
  w_n(\alpha)=\max_{p\in \Delta(K)}\{v_n(p)+p^T\alpha\}. \label{eq: w n}
\end{align}
Moreover, the initial regret $\hat{\alpha}_0$ of the primal game $\Gamma_n(p)$ is an optimal solution to the minimum problem (\ref{eq: v n}).
\end{theorem}
\begin{proof}
  First, we will show that
  \begin{align}
    v_n(p)\leq w_n(\alpha)-p^T\alpha, \forall p\in \Delta(K), \alpha\in\mathbb{R}^{|K|}. \label{eq: transient results 1}
  \end{align}
  Let $\tau^+$ be player 2's security strategy in the dual game $G_n(\alpha)$. We then have $w_n(\alpha)=\max_{p\in \Delta(K)}\max_{\sigma\in \Sigma_n} p^T\alpha+\gamma_n(p,\sigma,\tau^+)$, which implies that for any $p$ and $\alpha$,
  \begin{align*}
    &\max_{\sigma\in \Sigma_n} p^T\alpha+\gamma_n(p,\sigma,\tau^+)\\
    =&p^T\alpha+\max_{\sigma\in \Sigma_n} \gamma_n(p,\sigma,\tau^+)\\
     \leq &w_n(\alpha)
  \end{align*}
Hence, for any $p\in \Delta(K)$ and $\alpha\in \mathbb{R}^{|K|}$
\begin{align}
\max_{\sigma\in \Sigma_n} \gamma_n(p,\sigma,\tau^+) \leq w_n(\alpha)-p^T\alpha . \label{eq: transient results 3}
\end{align}
  Since for any $p\in \Delta(K)$, $v_n(p)\leq \max_{\sigma\in \Sigma_n} \gamma_n(p,\sigma,\tau^+)$, equation (\ref{eq: transient results 1}) is proven.

  Second, we show that for any $p\in \Delta(K)$, there exists an $\alpha\in \mathbb{R}^{|K|}$ such that
  \begin{align}
  v_n(p)\geq w_n(\alpha)-p^T\alpha. \label{eq: transient results 2}
  \end{align}
  Let $\tau^*$ be player 2's security strategy in the primal game $\Gamma_n(p)$. From the definition of the initial regret $\hat{\alpha}_0$ of the primal game $\Gamma_n(p)$, we see that $v_n(p)=-p^T \hat{\alpha}_0$. Notice that $\tau^*$ may not be player 2's security strategy any more if the initial probability changes. Therefore, we have for any $p'\in \Delta(K)$, $v_n(p')\leq \max_{\sigma\in\Sigma_n}\gamma_n(p',\sigma,\tau^*)=-p'^T \hat{\alpha}_0$.
  \begin{align*}
    w_n(\hat{\alpha}_0)=&\max_{p'\in \Delta(K)}\max_{\sigma\in \Sigma_n}\min_{\tau\in\mathcal{T}_n} p'^T \hat{\alpha}_0+\gamma_n(p',\sigma,\tau^*)\\
   =&\max_{p'\in \Delta(K)} p'^T \hat{\alpha}_0+\max_{\sigma\in \Sigma_n}\min_{\tau\in\mathcal{T}_n}\gamma_n(p',\sigma,\tau)\\
   =& \max_{p'\in \Delta(K)} p'^T \hat{\alpha}_0+ v_n(p').
  \end{align*}
  Since $v_n(p')\leq -p'^T \hat{\alpha}_0$ for any $p'\in \Delta(K)$, it can be derived that $w_n(\hat{\alpha}_0)\leq 0=v_n(p)+p^T\hat{\alpha}_0$,
  which proves that there exists an $\alpha\in\mathbb{R}^{|K|}$ such that equation (\ref{eq: transient results 2}) holds. Equation (\ref{eq: transient results 1}) and equation (\ref{eq: transient results 2}) imply equation (\ref{eq: v n}), and $\hat{\alpha}_0$ is an optimal solution to the minimum problem in (\ref{eq: v n}).

  Finally, we prove equation (\ref{eq: w n}).
  \begin{align*}
    w_n(\alpha)=&\min_{\tau\in \mathcal{T}_n} \max_{p\in \Delta(K)}\max_{\sigma\in \Sigma_n} p^T\alpha+\gamma_n(p,\sigma,\tau)\\
    =&\min_{\tau\in \mathcal{T}_n} \max_{p\in \Delta(K)} p^T\alpha+\max_{\sigma\in \Sigma_n}\gamma_n(p,\sigma,\tau).
  \end{align*}
  Function $\gamma_n(p,\sigma,\tau)$ is linear in $\tau$, and we have for any $\epsilon \in (0,1)$,
  \begin{align*}
  &\max_{\sigma\in \Sigma_n}\gamma_n(p,\sigma,\epsilon \tau+(1-\epsilon)\tau')\\
  =&\max_{\sigma\in \Sigma_n} \epsilon\gamma_n(p,\sigma, \tau)+(1-\epsilon) \gamma_n(p,\sigma,\tau') \\
  \leq & \epsilon\max_{\sigma\in \Sigma_n}\gamma_n(p,\sigma, \tau)+(1-\epsilon) \max_{\sigma\in \Sigma_n}\gamma_n(p,\sigma,\tau'),
  \end{align*}
  which shows that $\max_{\sigma\in \Sigma_n}\gamma_n(p,\sigma,\tau)$ is convex in $\tau$. Together with the fact that $p^T\alpha+\max_{\sigma\in \Sigma_n}\gamma_n(p,\sigma,\tau)$ is linear in $p$, according to the Sion's minimax theorem \cite{sion1958general}, we have
  \begin{align*}
    w_n(\alpha)=&\max_{p\in \Delta(K)} \min_{\tau\in \mathcal{T}_n} p^T\alpha+\max_{\sigma\in \Sigma_n}\gamma_n(p,\sigma,\tau)\\
    =& \max_{p\in \Delta(K)} p^T\alpha+v_n(p),
  \end{align*}
  which completes the proof.
\end{proof}

While Theorem \ref{theroem: game value relations} provides the relations between the game values of the primal game and the dual game, the next theorem states that the in some special case, the security strategy of the uninformed player in the dual game is also the security strategy of the uninformed player in the primal game.

\begin{theorem}
\label{theorem: relation in security strategy of uninformed player}
Given an optimal solution $\alpha^*$ to the minimum problem (\ref{eq: v n}), any security strategy of the uninformed player in the dual game $G_n(\alpha^*)$ is a security strategy of the uninformed player in the primal game $\Gamma_n(p)$.
\end{theorem}
\begin{proof}
 Let $\tau^+$ be a security strategy in the dual game $G_n(\alpha^*)$. Equation (\ref{eq: transient results 3}) implies that $\max_{\sigma\in \Sigma_n} \gamma_n(p,\sigma,\tau^+) \leq w_n(\alpha^*)-p^T\alpha^*=v_n(p)$. The last equality is derived from the fact that $\alpha^*$ is the optimal solution to the minimum problem (\ref{eq: v n}). Meanwhile, $\displaystyle v_n(p)\leq \max_{\sigma\in \Sigma_n} \gamma_n(p,\sigma,\tau^+)$. Therefore, we have $\displaystyle v_n(p)=\max_{\sigma\in \Sigma_n} \gamma_n(p,\sigma,\tau^+)$, and $\tau^+$ is a security strategy of the uninformed player in the primal game $\Gamma_n(p)$.
\end{proof}

Equation (\ref{eq: v n}) provides us a way to evaluate the game value of the primal game $\Gamma_n(p)$ from the game value of the dual game $G_n(\alpha)$. Generally speaking, for any initial vector payoff $\alpha$, $w_n(\alpha)-p^T \alpha$ is an upper bound on $v_n(p)$, and if we play the uninformed player's security strategy of the dual game in the primal game, the security level is no less than $v_n(p)$. However, if the initial vector payoff of the dual game is the initial regret in the primal game, we will have $v_n(p)=w_n(\hat{\alpha}_0)-p^T \hat{\alpha}_0$, and the security strategy of the uninformed player's in the dual game can guarantee an expected total payoff of $v_n(p)$ in the primal game. Therefore, when playing the primal game $\Gamma_n(p)$, we can see the game as a dual game $G_n(\hat{\alpha}_0)$ and play the dual game instead.

A security strategy of the uninformed player in the dual game has some nice properties. According to Theorem \ref{theorem: relation in security strategy of uninformed player}, these properties also apply to security strategies of the uninformed player in the primal game. To explore the properties of uninformed player's security strategy, we first present a recursive formula of the game value $w_n(\alpha)$ in dual game $G_n(\alpha)$.
\begin{proposition}
  \label{lemma: recursive formula transient result}
Consider a dual asymmetric information stochastic game $G_{n+1}(\alpha)$. Its game value $w_{n+1}(\alpha)$ satisfies
\begin{align}
  &w_{n+1}(\alpha) \label{eq: w n+1 to v n}\\
  =&\min_{\hat{y}\in \Delta(B)} \max_{\Pi\in \Delta(K\times A)} \sum_{k\in K,a\in A}\Pi(k,a)  (\alpha(k)+ M_k(a,:)\hat{y})  + \sum_{a\in A}\bar{x}(\Pi,a)v_n(\phi^T(\Pi,a)Q_a),\nonumber
\end{align}
where $v_n(p)$ is the game value of primal game $\Gamma_n(p)$.
\end{proposition}
\begin{proof}
According to equation (\ref{eq: w n}) and (\ref{eq: recursive formula primal 1}), we have
\begin{align*}
  &w_{n+1}(\alpha)\\
  =&\max_{p\in \Delta(K)}\max_{X\in \Delta(A)^{|K|}}\min_{\hat{y}\in \Delta(B)}p^T\alpha+ \sum_{k\in K} p(k)X^T(:,k) M_k \hat{y}+\sum_{a\in A}\bar{x}(p,X,a)v_n(\phi^T(p,s,a)Q_a) \\
  =& \max_{p\in \Delta(K)}\max_{X\in \Delta(A)^{|K|}} \min_{\hat{y}\in \Delta(B)}\left\{ \sum_{k\in K,a\in A}p(k) X(a,k) (\alpha(k)+ M_k(a,:)\hat{y}) \right.\\
  &\left.+ \sum_{a\in A}\bar{x}(p,X,a)v_n(\phi^T(p,s,a)Q_a)\right\}\\
  =& \max_{\Pi\in \Delta(K\times A)} \min_{\hat{y} \in \Delta(B)} \sum_{k\in K,a\in A}\Pi(k,a)  (\alpha(k)+ M_k(a,:)\hat{y}) +  \sum_{a\in A}\bar{x}(\Pi,a)v_n(\phi^T(\Pi,a)Q_a)
\end{align*}
where the last equality is derived by letting $\Pi(k,a)=p(k) X(a,k)$.

Next, we need to change the order of the maximum function and the minimum function. To this end, we will show that function $$f(\Pi,\hat{y})=\sum_{k\in K,a\in A}\Pi(k,a)  (\alpha(k)+ M_k(a,:)\hat{y}) + + \sum_{a\in A}\bar{x}(\Pi,a)v_n(\phi^T(\Pi,a)Q_a)$$ is concave in $\Pi$ and linear in $\hat{y}$. According to Lemma III.1 in \cite{li2014lp}, $\bar{x}(\Pi,a)v_n(\phi^T(\Pi,a)Q_a)=v_n(\Pi^T(:,a) Q_a)$. Since $v_n(p)$ is concave in $p$ \cite{zamir1992repeated}, it is concave in $\Pi$, and $f(\Pi,\hat{y})$ is also concave in $\Pi$. Together with the fact that $f(\Pi,\hat{y})$ is linear in $\hat{y}$, according to Sion's minimax theorem \cite{sion1958general}, we have equation (\ref{eq: w n+1 to v n}).
\end{proof}

The idea behind equation (\ref{eq: w n+1 to v n}) is similar to the idea of dynamic programming. The first term of (\ref{eq: w n+1 to v n}) is the expected current payoff, and the second term is the expected future payoff. In that expression, the uninformed player controls $y$, its strategy at stage 1, and aims to minimize the total expected payoff. The informed player controls $p$, the probability to choose the initial game state, and $X$, its strategy at stage 1, which means that it controls the joint probability $\Pi$ of the state and its action at stage 1, and the informed player's objective is to maximize the total expected payoff.

From stage 2 on, however, the game's state is not freely chosen by the informed player but distributed according to $\phi^T(\Pi,a)Q_a$. In turn, the future payoff can be seen as the value of primal game $\Gamma_n(\phi^T(\Pi,a)Q_a)$. Using Theorem \ref{theroem: game value relations}, we can further evaluate this value by, again, looking at the corresponding dual game. This, together with Proposition \ref{lemma: recursive formula transient result}, allows us to derive a recursive formula for the value of a dual game as follows.
\begin{proposition}
  \label{lemma: recursive formula and H B independent property}
Consider a dual asymmetric information stochastic game $G_{n+1}(\alpha)$. Its game value $w_{n+1}(\alpha)$ satisfies
\begin{align}
  &w_{n+1}(\alpha)\label{eq: recursive formula dual}\\
  =&\min_{\hat{y}\in \Delta(B)}\min_{(\beta_a \in \mathbb{R}^{|K|})_{ a\in A}} \max_{\pi\in \Delta(K\times A)} \sum_{k\in K, a\in A} \Pi(k,a) (M_k(a,:)\hat{y}+\alpha(k) -Q_a(k,:)\beta_a+w_n(\beta_a)),\nonumber
\end{align}
where $w_1(\alpha)=\min_{\hat{y}\in \Delta(B)} \max_{\Pi\in \Delta(K\times A)} \sum_{k\in K, a\in A} \Pi(k,a) (M_k(a,:)\hat{y}+\alpha(k))$.
\end{proposition}
\begin{proof}
From equation (\ref{eq: w n+1 to v n}) and (\ref{eq: v n}), we have
  \begin{align*}
  &w_{n+1}(\alpha)\nonumber\\
  =&\min_{\hat{y}\in \Delta(B)} \max_{\Pi\in \Delta(K\times A)} \left\{ \sum_{k\in K,a\in A}\Pi(k,a)  (\alpha(k)+ M_k(a,:)\hat{y}) \right.\\
  &\left.+ \sum_{a\in A}\bar{x}(\Pi,a) \min_{\beta_a\in \mathbb{R}^{|K|}}\{w_n(\beta_a)-\phi^T(\Pi,a)Q_a \beta_a\}\right\} \\
  =& \min_{\hat{y}\in \Delta(B)} \max_{\Pi\in \Delta(K\times A)} \min_{(\beta_a\in \mathbb{R}^{|K|})_{a\in A}}  \sum_{k\in K,a\in A}\Pi(k,a) \left(\alpha(k)+ M_k(a,:)\hat{y} -Q_a(k,:)\beta_a + w_n(\beta_a)\right).
\end{align*}

Now, we need to change the order of $\max_{\pi\in \Delta(K\times A)} \min_{\beta_a\in \mathbb{R}^{|K|},\forall a\in A}$. For this purpose, we need to show that $w_n(\cdot)$ is convex. Let $\beta_1,\beta_2$ be any $|K|$ dimensional real vectors. For any $\epsilon\in (0,1)$,
\begin{align*}
w_n(\epsilon\beta_1+(1-\epsilon)\beta_2)
=&\max_{p\in \Delta(K)}\{v_n(p)+p^T(\epsilon\beta_1+(1-\epsilon)\beta_2)\}\\
=&\max_{p\in \Delta(K)}\{\epsilon(v_n(p)+p^T\beta_1)+(1-\epsilon)(v_n(p)+p^T\beta_2)\}\\
\leq & \epsilon\max_{p\in \Delta(K)}\{v_n(p)+p^T\beta_1\}+(1-\epsilon)\max_{p\in \Delta(K)}\{v_n(p)+p^T\beta_2\}\\
=&\epsilon w_n(\beta_1)+(1-\epsilon)w_n(\beta_2).
\end{align*}
Therefore, $w_n(\cdot)$ is convex. Together with the fact that $\sum_{k\in K,a\in A}\Pi(k,a)(\alpha(k)$ $+ M_k(a,:)\hat{y} -Q_a(k,:)\beta_a + w_n(\beta_a))$ is linear in $\Pi$, according to Sion's minimax theorem \cite{sion1958general}, equation (\ref{eq: recursive formula dual}) is shown.
\end{proof}

The variable $\beta_a$ is introduced when we replace $v_n(\phi^T(\Pi,a)Q_a)$ with $$ \min_{\beta_a\in \mathbb{R}^{|K|}}\{w_n(\beta_a)-\phi^T(\Pi,a)Q_a \beta_a\}$$, and can be seen as the uninformed player's guess about the initial regret, i.e. future cost compared to zero, of the future primal game $\Gamma_n(\phi^T(\Pi,a)Q_a)$. On one hand, the uninformed player controls its strategy $y$. On the other hand, it takes a guess about the initial regret $\beta_a$ of the future primal game given the informed player's current action $a$. Given $y$ and $\beta_a$, the expected total payoff is $\sum_{k\in K, a\in A} \Pi(k,a) (M_k(a,:)\hat{y}+\alpha(k) -Q_a(k,:)\beta_a+w_n(\beta_a))$. Since the informed player aims to maximize the expected total payoff, the uninformed player will choose $y$ and $\beta_a$ such that the maximum expected total payoff is minimized. The optimal solution $(\beta_a^*)_{a\in A}$ to the minmax problem (\ref{eq: recursive formula dual}) is  called the regret at stage $2$ in a dual game $G_{n+1}(\alpha)$. The regret at stage $t$ in a dual game is formally defined as below.
\begin{definition}
  \label{definition: regret, dual}
  Consider a dual game $G_N(\alpha)$. We call $\alpha$ the regret at stage 1, and denote it as $\alpha_1$.

  Given regret $\alpha_t$ and informed player's action $a$ at stage $t$, let $\hat{y}^*$ and $(\beta_a^*)_{a\in A}$ be the optimal solution to the following problem.
  \begin{align}
  \min_{\hat{y}\in \Delta(B)}\min_{(\beta_a \in \mathbb{R}^{|K|})_{ a\in A}} \max_{\pi\in \Delta(K\times A)} \sum_{k\in K, a\in A} \Pi(k,a) (M_k(a,:)\hat{y}+\alpha_t(k) -Q_a(k,:)\beta_a+w_{N-t}(\beta_a)), \label{eq: regret update, dual}
  \end{align}
  We call $\beta_a^*$ the regret at stage $t+1$, and denote it as $\alpha_{t+1}$.
\end{definition}

Now, we are ready to present a regret based strategy for the uninformed player in a dual game. Let $\hat{y}^*$ and $\beta^*$ be the optimal solution to the minmax problem (\ref{eq: recursive formula dual}). At every stage $t$ with vector payoff $\alpha_t$, we can use the optimal solution $\hat{y}^*$ as the current strategy of the uninformed player, and update the vector payoff $\alpha_{t+1}$ at the next stage to $\beta^*_{a_t}$. The detailed algorithm is given below.
\begin{algorithm}\hfill{}
\label{algorithm: vector based security strategy, dual}
\begin{enumerate}
  \item Initialization
    \begin{itemize}
      \item Read payoff matrices $M$, transition matrices $Q$, and initial vector payoff $\alpha$.
      \item Set stage $t=1$, and $\alpha_t=\alpha$.
    \end{itemize}
  \item Find out the optimal solution $\hat{y}^*$ and $(\beta_a^*)_{a\in A}$ to the minmax problem (\ref{eq: regret update, dual}).
  \item Draw an action according to $\hat{y}^*$ and read the action $a_t$ of the informed player.
  \item Set $t=t+1$, and update $\alpha_t=\beta^*_{a_t}$.
  \item If $t\leq N$, go to step 2). Otherwise, end.
\end{enumerate}
\end{algorithm}
We shall notice that the strategy derived in this algorithm is independent of $\sigma$, the strategy of the informed player.

The next question is whether the strategy constructed in Algorithm \ref{algorithm: vector based security strategy, dual} is a security strategy of the uninformed player in dual game $G_N(\alpha)$, and the answer is yes. We provide the detail in the following theorem.
\begin{theorem}
  \label{theorem: regret based security strategy, dual game}
In the dual asymmetric information game $G_N(\alpha)$, the uninformed player has a security strategy at stage $t$ that only depends on regret $\alpha_t$ and the stage $t$. Moreover, such a security strategy can be constructed using Algorithm \ref{algorithm: vector based security strategy, dual}.
\end{theorem}
\begin{proof}
We will construct a strategy of the uninformed player as in Algorithm \ref{algorithm: vector based security strategy, dual}.

Mathematical induction is used to prove the theorem. First, we see that in dual game $G_1(\alpha)$, it is easy to check that $\hat{y}^*$ is a security strategy of player 2.

Assume that $\tau^*(\alpha)$ is a security strategy of player 2 in dual game $G_n(\alpha)$ constructed as in Algorithm \ref{algorithm: vector based security strategy, dual}. Let $\hat{y}^*$ and $\beta^*$ be the optimal solution to the minmax problem (\ref{eq: recursive formula dual}). We will show that in dual game $G_{n+1}(\alpha)$, if player 2 plays $\hat{y}^*$ at stage 1, and $\tau^*(\beta^*_{a_1})$ afterwards, then player 2 can guarantee the game value $w_{n+1}(\alpha)$.

For any $a\in A$, since $\tau^*(\beta^*_a)$ is a security strategy of player 2 in dual game $G_n(\beta^*_a)$, we have
\begin{align*}
  \gamma_n(p,\sigma,\tau^*(\beta^*_a))+p^T\beta^*_a\leq w_n(\beta^*_a), \forall p\in \Delta(K),\forall \sigma\in \Sigma_n.
\end{align*}
Let $p_{\Pi}(k')=\sum_{k\in K,a\in A}\Pi(k,a) Q_a(k,k')$, where $\Pi\in \Delta(K\times A)$. We have
\begin{align*}
  \gamma_n(p_{\Pi},\sigma,\tau^*(\beta^*_a)) \leq \sum_{k\in K,a\in A}\Pi(k,a)( w_n(\beta^*_a)-\sum_{k'\in K}Q^a_{k,k'}\beta_a^*(k')),
\end{align*}
for all $\Pi\in \Delta(K\times A), a\in A, \sigma\in \Sigma_n $. Hence,
\begin{align*}
  \sum_{k\in K,a\in A}\Pi(k,a) \gamma_n(p_{\Pi},\sigma,\tau^*(\beta^*_a)) \leq \sum_{k\in K,a\in A}\Pi(k,a)( w_n(\beta^*_a)-Q_a(k,:)\beta_a^{*}),
\end{align*}
for all $\Pi\in \Delta(K\times A), \sigma\in \Sigma_n$. Hence,
\begin{align*}
 &\sum_{k\in K,a\in A}\Pi(k,a) (M_k(a,:) \hat{y}^*+\alpha(k)+\gamma_n(p_{\Pi},\sigma,\tau^*(\beta^*_a))) \\
 \leq &\sum_{k\in K,a\in A}\Pi(k,a)( w_n(\beta^*_a)-Q_a(k,:)\beta_a^*+M_k(a,:) \hat{y}^*+\alpha(k)),
\end{align*}
for all $\Pi\in \Delta(K\times A), a\in A, \sigma\in \Sigma_n$. Therefore, we have for all $\pi\in \Delta(K\times A), \sigma\in \Sigma_n$,
\begin{align*}
&\sum_{k\in K,a\in A}\Pi(k,a) (M_k(a,:) \hat{y}^*+\alpha(k)+\gamma_n(p_{\Pi},\sigma,\tau^*(\beta^*_a)))  \\
\leq &\max_{\pi\in \Delta(K\times A)} \sum_{k\in K,a\in A}\Pi(k,a)( w_n(\beta^*_a)-Q_a(k,:)\beta_a^*+M_k(a,:) \hat{y}^*+\alpha(k)).
\end{align*}
Since $\hat{y}^*$ and $\beta^*$ is the optimal solution to the minmax problem (\ref{eq: recursive formula dual}), we have for all $\Pi\in \Delta(K\times A), \sigma\in \Sigma_n$,
\begin{align*}
\sum_{k\in K,a\in A}\Pi(k,a) (M_k(a,:) \hat{y}^*+\alpha(k)+\gamma_n(p_{\Pi},\sigma,\tau^*(\beta^*_a))) \leq w_{n+1}(\alpha).
\end{align*}
Let $\Pi(k,a)=p_1(k) X(a,k)$, where $p_1$ is player 1's strategy to choose a state, and $X(:,k)=\sigma_1(k,\emptyset)$ is player 1's strategy to choose an action at stage 1 given state $k$. It is straight forward to show that
$$g_{n+1}(\alpha,(X,\sigma),(\hat{y}^*,\tau^*(\beta^*_a))=\sum_{k\in K,a\in A}\Pi(k,a) (M_k(a,:) \hat{y}^*+\alpha(k)+\gamma_n(p_{\Pi},\sigma,\tau^*(\beta^*_a))).$$
Therefore, we have for any $p_1\in \Delta(K)$ and $(x,\sigma)\in \Sigma_{n+1}$, $$g_{n+1}(\alpha,(X,\sigma),(\hat{y}^*,\tau^*(\beta^*_a))) \leq w_{n+1}(\alpha),$$ which completes the proof.
\end{proof}

Regret $\alpha_t$ and stage $t$ form a sufficient statistics for the uninformed player to make decisions. The reason why $\alpha_t$ plays such a role (similar as that of $p_t$ for the informed player in the primal game) for the uninformed player in the dual game can be explained intuitively as follows. In a primal game, to decide the current strategy, the informed player needs to estimate the current payoff and the future payoff $v_{n-1}(p_{t+1})$ (equation (\ref{eq: recursive formula primal 1})). The belief $p_{t+1}$ at the next stage decides the future payoff. Similarly, in a dual game, to decide the current strategy, the uninformed player also needs to estimate the current payoff and the future payoff $v_n(p_{t+1})$ (equation (\ref{eq: w n+1 to v n})). While the uninformed player cannot compute $p_{t+1}$ without the informed player's strategy, it can consider the worst case scenario, and compute the worst case vector security level of the uninformed player, and hence the regret for the future game at the next stage. The regret at the next stage characterizes the worst case future payoff, and hence plays an important role for the uninformed player in making decisions.

Noticing that regret $\alpha_{t+1}$ at stage $t+1$ derived in Algorithm \ref{algorithm: vector based security strategy, dual} only depends on regret $\alpha_{t}$ at stage $t$ and the player 1's action $a_t$, we have the following corollary.
\begin{corollary}
  \label{corollary: H B independent strategy, uninformed, dual game}
In a dual game $G_N(\alpha)$, the uninformed player has a security strategy that, at stage $t$, only depends on stage $t$ and the action history $I_t$ of the informed player.
\end{corollary}

Now, let us get back to the primal game $\Gamma_N(p)$. Since for any $p\in \Delta(K)$, there always exists an $\alpha$ such that any security strategy of the uninformed player in dual game $G_N(\alpha)$ is a security strategy of the uninformed player in primal game $\Gamma_N(p)$ (Theorem \ref{theroem: game value relations}), the properties described in Theorem \ref{theorem: regret based security strategy, dual game} and Corollary \ref{corollary: H B independent strategy, uninformed, dual game} are also true for security strategies of the uninformed player in the primal game. Let us first define regrets in a primal game.
\begin{definition}
\label{definition: regret}
In a primal game $\Gamma_N(p)$, $\hat{\alpha}_1\in \mathbb{R}^{|K|}$ is called the regret at stage 1 if $\hat{\alpha}_1$ is an optimal solution to the minimum problem in equation (\ref{eq: v n}).

Given the regret $\hat{\alpha}_t$ and the informed player's action $a_t$ at stage $t$, let $\hat{y}^*$ and $\beta^*$ be an optimal solution of the minmax problem in (\ref{eq: regret update, dual}).
We call $\beta^*_{a_t}$ the regret at stage $t+1$, and denote it as $\hat{\alpha}_{t+1}$.
\end{definition}

According to Definition \ref{definition: regret}, the initial regret $\hat{\alpha}_0$ defined in (\ref{eq: initial regret}) is also the regret at stage 1 in primal game $\Gamma_N(p)$.

Theorem \ref{theroem: game value relations} says that any security strategy of the uninformed player in the corresponding dual game with initial vector payoff $\hat{\alpha}_1$, the regret at stage 1 in the primal game, is a security strategy of the uninformed player in the primal game. Theorem \ref{theorem: regret based security strategy, dual game} states that if we update the regret in dual game as in Algorithm \ref{algorithm: vector based security strategy, dual}, then the security strategy of the uninformed player depends only on the regret and the current stage. The regret defined in Definition \ref{definition: regret} in primal game is updated exactly the same as in Definition \ref{definition: regret, dual}. Therefore, we can compute uninformed player's security strategy and update the regret using the following algorithm.

\begin{algorithm}\hfill{}
  \label{algorithm: regret based security strategy, primal}
\begin{itemize}
  \item Initialization
    \begin{itemize}
      \item Read payoff matrices $M$, transition matrices $Q$, time horizon $N$, and initial probability $p_0$.
      \item Set $t=1$, and $p_t=p_0$.
    \end{itemize}
  \item Compute the optimal solution $\alpha^*$ to the minimum problem in (\ref{eq: v n}) with $n=N$ and $p=p_0$. Set $\hat{\alpha}_1=\alpha^*$.
  \item Compute the optimal solution $\hat{y}^*$ and $(\beta_a^*)_{a\in A}$ to the minmax problem (\ref{eq: regret   update, dual}) with $\alpha_t=\hat{\alpha}_t$.
  \item Draw an action $b\in B$ according to $\hat{y}^*$ and read the action $a$ of the informed player.
  \item Set $t=t+1$, and update $\hat{\alpha}_t=\beta^*_a$.
  \item If $t\leq N$, go to step 3. Otherwise, end.
\end{itemize}
\end{algorithm}

\begin{corollary}
  \label{corollary: regret based, primal, uninformed}
  Consider a primal game $\Gamma_N(p)$. The uninformed player has a security strategy at stage $t$, that only depends on the regret $\hat{\alpha}_t$ and the stage $t$, and such a security strategy can be constructed by Algorithm \ref{algorithm: regret based security strategy, primal}. Moreover, this security strategy only depends on stage $t$ and the history action $I_t$ of the informed player.
\end{corollary}


Now that the basic steps to find out a security strategy of the uninformed player in a primal game are clear, we further provide LP formulations to compute the regret at stage 1, a security strategy in the corresponding dual game, and the regret at every stage.

\subsection{LP formulation of history based security strategies}
As mentioned in Section \ref{subsec: history, informed}, a security strategy of the uninformed player can also be computed by solving an LP based on a sequence form \cite{koller1996efficient} whose size is linear with respect to the size of the game tree ($|K|^N\times|A|^N\times|B|^N$). Based on the fact that player 2 has a security strategy that is independent of its own history action, we can reduce the size of the LP to be linear with respect to $|B|$. This simplified LP formulation is introduced in this subsection.

Consider an asymmetric information stochastic game $\Gamma_N(p_0)$. We define a realization plan $r_t(S_t,I_{t+1})$ given state history $S_{t}$ and player 1's action history as $I_{t+1}$ as $r_t(S_t,I_{t+1})=p_0(k_1)\prod_{s=1}^t \sigma_s^{a_s}(k_s,I_s)$, where $\sigma_s^{a_s}(k_s,I_s)$ denotes the $a_s$th element of $\sigma_s(k_s,I_s)$ . Let $r=(r_t)_{t=1}^N$, and $R$ be the set of all possible values that the realization plan can take. The realization plan satisfies
\begin{align}
  &\sum_{a_t}r_t(S_t,I_{t+1})=r_{t-1}(s_{t-1},I_t), \forall S_t\in K^t, I_t\in A^{t-1}, \forall t=1,\ldots,N, \label{eq: realization 1}\\
  &r_t(S_t,I_{t+1})\geq 0, \forall t=1,\ldots,N. \label{eq: realization 2}
\end{align}
It is straight forward to show that $p(S_N,I_{N+1})=r_N(S_N,I_{N+1})\prod_{t=1}^{N-1}Q_{a_{t}}(k_t,k_{t+1})$.

Before presenting the simplified LP formulation for uninformed players, we would like to introduce some variables used in the LP formulation. Based on Corollary \ref{corollary: regret based, primal, uninformed}, we only consider player 2's strategies that depend on the informed player's history action only.
Let $F_t=(S_t,I_t)\in K^t\times A^{t-1}$ be the full information the informed player has at the beginning of stage $t$. Let $y_{I_t}=\tau_t(I_t)\in \Delta(B)$, $y=(y_{I_t})_{I_t\in A^{t-1}, t=1,\ldots,N}$, and $Y$ be the set of all possible values that y can take. Let $\ell_{F_t}\in \mathbb{R}$ be a real variable. We use $\ell=(\ell_{F_t})_{F_t\in K^t\times A^{t-1},t=1,\ldots,N}$ to denote the collection of the $\ell$ variable, and use $L$ to denote the set of all possible values that $l$ can take. The simplified LP formulation is given below.
\begin{theorem}
  \label{theorem: LP, history based, primal, uninformed}
Consider a primal asymmetric information stochastic game $\Gamma_n(p)$. Its game value $v_n(p)$ satisfies
\begin{align}
  v_n(p)=&\min_{y\in Y}\min_{\ell\in L} \sum_{k\in K}p^k \ell_{F_1}, &&where\ F_1=((k,\emptyset)), \label{eq: LP, uninformed, primal 1}\\
  s.t.& \sum_{k_t\in K} \ell_{F_t}\leq \ell_{F_{t-1}}, &&\forall F_t\in K^t\times A^{t-1}, \forall t=2,\ldots,n, \nonumber\\
                                                    &&& where\ F_t=(F_{t-1},(k_t,a_{t-1})) \label{eq: LP, uninformed, primal 2}\\
  &\prod_{t=1}^{n-1}Q_{a_t}(k_t,k_{t+1}) \sum_{t=1}^n M_{k_t}(a_t,:) y_{I_t} \leq \ell_{F_n},& &\forall F_n\in K^n\times A^{n-1}, \forall a_n\in A. \nonumber \\
  &&& where\ F_n=((k_1,\emptyset),\ldots,(k_n,a_{n-1}))\label{eq: LP, uninformed, primal 3}
\end{align}
A security strategy $\tau^*_t(I_t)$ of player 2 at stage $t$ given the action history $I_t$ of player 1 is $\tau^*_t(I_t)=y^*_{I_t}$, where $y^*$ is the optimal solution of LP (\ref{eq: LP, uninformed, primal 1}-\ref{eq: LP, uninformed, primal 3}).
\end{theorem}
\begin{proof}
\begin{align*}
  v_n(p)=&\min_{\tau\in \mathcal{T}_n}\max_{\sigma\in \Sigma_n}\sum_{S_n\in K^n, I_{n+1}\in A^n} p(S_n,I_{n+1})E(\sum_{t=1}^n M_{k_t}(a_t,b_t)|S_n,I_{n+1})\\
  =&\min_{y\in Y}\max_{\sigma\in \Sigma_n} \sum_{S_n\in K^n, I_{n+1}\in A^n} p(k_1)\prod_{t=1}^n \sigma_t^{a_t}(k_t,I_t)\prod_{t=1}^{n -1} Q_{a_t}(k_t,k_{t+1})\sum_{t=1}^n M_{k_t}(a_t,:) y_{I_t}\\
  =&\min_{y\in Y}\max_{r\in R} \sum_{S_n\in K^n, I_{n+1}\in A^n} r(S_n,I_{n+1})\left(\prod_{t=1}^{n -1} Q_{a_t}(k_t,k_{t+1})\sum_{t=1}^n M_{k_t}(a_t,:) y_{I_t}\right),\\
  s.t. & equation (\ref{eq: realization 1}-\ref{eq: realization 2}).
\end{align*}
According to the strong duality theorem, equation (\ref{eq: LP, uninformed, primal 1}-\ref{eq: LP, uninformed, primal 3}) is shown.
\end{proof}

The size of the LP problem in (\ref{eq: LP, uninformed, primal 1}-\ref{eq: LP, uninformed, primal 3}) is $O(|K|^n|A|^n|B|)$. Let us first look at the variable size. Variable $y$ has a size of $(1+|A|+\ldots+|A|^{n-1})|B|$ which is of order $|A|^n|B|$. Variable $\ell$ is of size $|K|(1+|A||K|+\ldots+(|A||K|)^{n-1})$ which is of order $|A|^n|K|^n$. Next, we will analyze the constraint size. Constraint \ref{eq: LP, uninformed,   primal 2} has a size of $(|A||K|+(|A||K|)^2+\ldots+(|A||K|)^{n-1}$ which is of order $|A|^n|K|^n$. The size of constraint \ref{eq: LP, uninformed, primal 3} is also of order $|A|^n|K|^n$. Therefore, in all, we see that the size of the LP problem (\ref{eq: LP, uninformed, primal 1}-\ref{eq: LP, uninformed, primal 3}) is $O(|K|^n|A|^n|B|)$.

The LP formulation provides us not only with a history based security strategy for player 2, but also the regret $\hat{\alpha}_1$ at stage 1 in the primal game.
\begin{proposition}
  \label{lemma: initial regret}
  Let $y^*,\ell^*$ be the optimal solution of LP problem (\ref{eq: LP, uninformed, primal 1}-\ref{eq: LP, uninformed, primal 3}). The initial regret of the primal game $\Gamma_n(p)$ is $-\ell^*_1$, where $\ell_1^*=(\ell^*_{F_1})_{F_1\in K\times \emptyset}$, i.e. $-\ell^*_1$ is an optimal solution to the minimum problem in equation (\ref{eq: v n}).
\end{proposition}
\begin{proof}
First, we show that $w_n(-\ell^*_1) \geq 0$. Equation (\ref{eq: w n}) indicates that $w_n(-\ell^*_1)\geq v_n(p)-p^T \ell^*_1=0$.

Second, we show that $w_n(-\ell^*_1)\leq 0$. To this end, we first show that $\bar{\alpha}_{y^*}(k)=\mathbf{E}(\sum_{t=1}^n M_{k_t}(a_t,b_t)|k_1=k)\leq \ell^*_1(k)$, for all $k\in K$.
\begin{align*}
  \bar{\alpha}_{y^*}(k)
  =&\sum_{k_{2:n}\in K^{n-1}}\sum_{a_{1:n}\in A^n}P(k_2,\ldots,k_n,a_1,\ldots,a_n|k_1=k)\\
  &\mathbf{E}\left(\sum_{t=1}^n M_{k_t}(a_t,b_t))|k_1=k,k_2,\ldots,k_n,a_1,\ldots,a_n)\right)\\
  =&\sum_{k_{2:n}\in K^{n-1}}\sum_{a_{1:n}\in A^n}\prod_{t=1}^n\sigma_t^{a_t}(k_t,I_t)\prod_{t=1}^{n-1}Q_{a_{t}}(k_t,k_{t-1}) \sum_{t=1}^nM_{k_t}(a_t,:)y^*_{I_t}.
\end{align*}
Equation (\ref{eq: LP, uninformed, primal 3}) implies that
\begin{align*}
  \bar{\alpha}_{y^*}(k)\leq &\sum_{k_{2:n}\in K^{n-1}}\sum_{a_{1:n}\in A^n}\prod_{t=1}^n\sigma_t^{a_t}(k_t,I_t)\ell_{F_n}\\
  &where\ F_n=((k,\emptyset),(k_2,a_1),\ldots,(k_n,a_{n-1})) \\
  =&\sum_{k_{2:n}\in K^{n-1}}\sum_{a_{1:n-1}\in A^{n-1}}\prod_{t=1}^{n-1}\sigma_t^{a_t}(k_t,I_t)\ell_{F_n}\\
  \leq & \sum_{k_{2:n-1}\in K^{n-2}}\sum_{a_{1:n-1}\in A^{n-1}}\prod_{t=1}^{n-1}\sigma_t^{a_t}(k_t,I_t) l_{F_{n-1}}
\end{align*}
where the last inequality is derived from equation (\ref{eq: LP, uninformed,   primal 2}). Following the same steps, we can show that $\bar{\alpha}_{y^*}(k)\leq \ell^*_1(k)$. For any $p'\in \Delta(K)$, we have $\displaystyle v_n(p')\leq \max_{\sigma\in \Sigma_n}\gamma_n(p',\sigma,y^*) =p'^T \bar{\alpha}_{y^*}\leq p'^T \ell^*_1$. Therefore, $w_n(-\ell^*_1)=\max_{p'\in \Delta(K)} v_n(p')-p'^T \ell^*_1 \leq 0$.

Therefore, we have $w_n(-\ell^*_1)=0$. From equation (\ref{eq: LP, uninformed, primal 1}), we have $w_n(-\ell^*_1)-p^T(-\ell^*_1)=v_n(p)$. This completes the proof.
\end{proof}

\subsection{LP formulation of regret based security strategy}
As discussed in Section \ref{subsection: strategy property, uninformed}, the uninformed player can construct a regret based security strategy following Algorithm \ref{algorithm: regret based, primal}. Proposition \ref{lemma: initial regret} provides an LP formulation to compute the regret at stage 1. In this section, we further study this LP formulation and show how the regret vector can be efficiently updated. To this end, we first introduce the LP formulation to compute the game value of a dual game $G_n(\alpha)$.

\begin{proposition}
\label{lemma: LP, dual, history}
The game value $w_n(\alpha)$ of a dual game $G_n(\alpha)$ satisfies
\begin{align}
  w_n(\alpha)=& \min_{y\in Y}\min_{\ell\in L}\min_{\hat{\ell}\in\mathbb{R}} \hat{\ell} \label{eq: LP, uninformed, dual 1} \\
  s.t.& \alpha(k)+\ell_{F_1} \leq \hat{\ell}, &&\forall k\in K, where\ F_1=((k,\emptyset)), \label{eq: LP, uninformed, dual 2}\\
  & \sum_{k_t\in K} \ell_{F_t}\leq l_{F_{t-1}},&& \forall F_t\in K^t\times A^{t-1}, \forall t=2,\ldots,n,\nonumber\\
  &&& where F_t=(F_{t-1},(k_t,a_{t-1})) \label{eq: LP, uninformed, dual 3}\\
  &\prod_{t=1}^{n-1}Q_{a_t}(k_t,k_{t+1}) \sum_{t=1}^n M_{k_t}(a_t,:) y_{I_t} \leq \ell_{F_n},& &\forall F_n\in K^n\times A^{n-1}, \forall a_n\in A. \label{eq: LP, uninformed, dual 4}
\end{align}
\end{proposition}
\begin{proof}
  \begin{align*}
    w_n(\alpha)=&\min_{y\in Y}\max_{p\in \Delta(K)}\max_{\sigma\in \Sigma_n}p^T\alpha+\mathbf{E}(\gamma_n(p,\sigma,\tau))\\
    =&\min_{y\in Y}\max_{p\in \Delta(K)}p^T\alpha+\max_{\sigma\in \Sigma_n}\mathbf{E}(\gamma_n(p,\sigma,\tau))
  \end{align*}
Following the same steps as in the proof of Theorem \ref{theorem: LP, history based, primal, uninformed}, we have
\begin{align*}
  w_n(\alpha)=& \min_{y\in Y}\max_{p\in \Delta(K)}p^T\alpha+ \min_{\ell\in L} \sum_{k\in K}p^k \ell_{F_1}, where\ F_1=((k,\emptyset))\\
  s.t. equation (\ref{eq: LP, uninformed,   primal 2}-\ref{eq: LP, uninformed, primal 3}).
\end{align*}
Since $p^T\alpha+\sum_{k\in K}p^k \ell_{F_1}$ is linear in both $p$ and $l$, according to Sion's minimax theorem \cite{sion1958general}, we have
\begin{align*}
  w_n(\alpha)=&\min_{y\in Y}\min_{\ell\in L} \max_{p\in \Delta(K)}p^T\alpha+ \sum_{k\in K}p^k \ell_{F_1}\\
  s.t. equation (\ref{eq: LP, uninformed,   primal 2}-\ref{eq: LP, uninformed, primal 3}).
\end{align*}
According to the strong duality theorem, equation (\ref{eq: LP, uninformed, dual 1}-\ref{eq: LP, uninformed, dual 4}) is shown.
\end{proof}

Now, we are ready to present the LP formulation to compute the regret based security strategy of player 2 and to update the regret in a primal game. With a little abuse of notation $y$ and $\ell$, we use $y_{a,I_t}\in \Delta(B)$ to indicate a $|B|$ dimensional probability variable given $a\in A$ and $I_t\in A^{t-1}$, and $\ell_{a,F_t}\in \mathbb{R}$ to denote a scalar variable given $a\in A$ and $F_t\in K^t\times A^{t-1}$. The collection of $(y_{a,I_t})){I_t\in A^{t-1},t=1,\ldots,n}$ is denoted as $y_a$, and the collections of $(\ell_{a,F_t})_{F_t\in K^t\times A^{t-1},t=1,\ldots,n}$ is denoted by $\ell_a$.
\begin{theorem}
Consider a primal game $\Gamma_{n+1}(p)$. Let $\hat{\alpha}_1$ be the regret at stage 1 in the primal game. The game value $w_{n+1}(\hat{\alpha}_1)$ of the dual game $G_{n+1}(\hat{\alpha}_1)$ satisfies
\begin{align}
  w_{n+1}(\hat{\alpha}_1)=&\min_{\hat{y}\in \Delta(B)}\min_{(\beta_a\in \mathbb{R}^{|K|})_{a\in A}}\min_{\tilde{\ell}\in \mathbb{R}}\min_{(\hat{\ell}_a\in \mathbb{R})_{a\in A} }\min_{(y_a\in Y,\ell_a\in L)_{a\in A}} \tilde{\ell} \label{eq: LP, regret based, dual 1}\\
  s.t.& M_k(a,:)\hat{y}+\hat{\alpha}_1(k) -Q_a(k,:)\beta_a+\hat{\ell}_a \leq \tilde{\ell}, \forall a\in A, k\in K, \label{eq: LP, regret based, dual 2}\\
  & \beta_a(k)+\ell_{a,F_1} \leq \hat{\ell}_a, \forall a\in A, k\in K, where\ F_1=((k,\emptyset)) \label{eq: LP, regret based, dual 3}\\
  & \sum_{k_t\in K} \ell_{a,F_t}\leq \ell_{a,F_{t-1}},  \forall a\in A,\forall F_t\in K^t\times A^{t-1},\forall t=2,\ldots,n,\nonumber \\
  & where\ F_t=(F_{t-1},(k_t,a_{t-1})) \label{eq: LP, regret based, dual 4}\\
  &\prod_{t=1}^{n-1}Q_{a_t}(k_t,k_{t+1}) \sum_{t=1}^n M_{k_t}(a_t,:) y_{a,I_t} \leq \ell_{a,F_n}, \forall a\in A, F_n\in K^n\times A^{n-1}, \forall a_n\in A.\label{eq: LP, regret based, dual 5}
\end{align}
Moreover, player 2's security strategy at the current step is $\hat{y}^*$ and the regret at the next step is $\beta^*_a$ if the current action of player 1 is $a$.
\end{theorem}
\begin{proof}
  Theorem \ref{theroem: game value relations} indicates that any security strategy of player 2 in dual game $G_{n+1}(\hat{\alpha}_1)$ is a security strategy of player 2 in the primal game $\Gamma_{n+1}(p)$.

  Let $\hat{y}^*$ and $\beta^*$ be the optimal solution to the minmax problem (\ref{eq: recursive formula dual}). According to Theorem \ref{theorem: regret based security strategy, dual game} and Definition \ref{definition: regret}, the current security strategy of player 2 in dual game $G_{n+1}(\hat{\alpha}_1)$ is the optimal solution $\hat{y}^*$, and the regret at the next stage is $\beta_a^*$ if the current action of player 1 is $a$. Now we need to build an LP to solve the minmax problem (\ref{eq: recursive formula dual}).

  According to the strong duality theorem, we have
  \begin{align}
    &\max_{\pi\in \Delta(K\times A)} \Pi(k,a)( M_k(a,:)\hat{y}+\hat{\alpha}_1(k) -Q_a(k,:)\beta_a+w_n(\beta_a)  )\\
    =&\min_{\tilde{\ell}\in\mathbb{R}} \tilde{\ell} \label{eq: -1}\\
    s.t.& M_k(a,:)\hat{y}+\hat{\alpha}_1(k) -Q_a(k,:)\beta_a+w_n(\beta_a) \leq \tilde{\ell}, \forall a\in A,k\in K \label{eq: -2}\\
    =& \min_{\tilde{\ell}\in\mathbb{R}} \tilde{\ell} \label{eq: 0}\\
    s.t.& M_k(a,:)\hat{y}+\hat{\alpha}_1(k) -Q_a(k,:)\beta_a+\min_{\hat{\ell}_a\in\mathbb{R}}\min_{(y_a\in Y,\ell_a\in L)_{a\in A}} \hat{\ell}_a \leq \tilde{\ell}, \forall a\in A,k\in K,  \label{eq: 1}
  \end{align}
  where $\hat{\ell},y_a,\ell_a, \tilde{\ell}$ satisfies
  \begin{align}
  & \beta_a(k)+\ell_{a,F_1} \leq \hat{\ell}_a, \forall a\in A, k\in K, where\ F_1=((k,\emptyset)) \label{eq: 2}\\
  & \sum_{k_t\in K} \ell_{a,F_t}\leq \ell_{a,F_{t-1}},  \forall a\in A,\forall F_t\in K^t\times A^{t-1},\forall t=2,\ldots,n,\nonumber \\
  & where\ F_t=(F_{t-1},(k_t,a_{t-1})) \label{eq: 3}\\
  &\prod_{t=1}^{n-1}Q_{a_t}(k_t,k_{t+1}) \sum_{t=1}^n M_{k_t}(a_t,:) y_{a,I_t} \leq \ell_{a,F_n}, \forall a\in A, F_n\in K^n\times A^{n-1}, \forall a_n\in A.\label{eq:4}
  \end{align}

  This is a nested LP. We will show that the optimal value of the embedded LP is the same as the the optimal value of the following LP.
  \begin{align}
   &\min_{\tilde{\ell}\in\mathbb{R}}\min_{(\hat{\ell}_a\in\mathbb{R})_{a\in A}}\min_{(y_a\in Y,\ell_a\in L)_{a\in A}} \tilde{\ell} \label{eq: a1}\\
    s.t.& M_k(a,:)\hat{y}+\hat{\alpha}_1(k) -Q_a(k,:)\beta_a+\hat{\ell}_a \leq \tilde{\ell}, \forall a\in A,k\in K, \label{eq: a2}\\
       & \beta_a(k)+\ell_{a,F_1} \leq \hat{\ell}_a, \forall a\in A, k\in K, where\ F_1=((k,\emptyset))\label{eq: a3}\\
  & \sum_{k_t\in K} \ell_{a,F_t}\leq \ell_{a,F_{t-1}},  \forall a\in A,\forall F_t\in K^t\times A^{t-1},\forall t=2,\ldots,n,\nonumber \\
  & where\ F_t=(F_{t-1},(k_t,a_{t-1})),\label{eq: a4}\\
  &\prod_{t=1}^{n-1}Q_{a_t}(k_t,k_{t+1}) \sum_{t=1}^n M_{k_t}(a_t,:) y_{a,I_t} \leq \ell_{a,F_n}, \forall a\in A, F_n\in K^n\times A^{n-1}, \forall a_n\in A. \label{eq: a5}
  \end{align}

  Let $\tilde{l}^*, (\hat{l}^*_a)_{a\in A}, (y^*_a)_{a\in A}, (l_a^*)_{a\in A}$ be the optimal solution to (\ref{eq: 0}-\ref{eq:4}), and $\tilde{l}^+,$ $ (\hat{l}^+_a)_{a\in A},$ $ (y^+_a)_{a\in A}, (l_a^+)_{a\in A}$ be the optimal solution to (\ref{eq: a1}-\ref{eq: a5}). We first show that $\tilde{l}^*\leq \tilde{l}^+$. Since for any $a\in A$, $\hat{l}_a^+,y^+_a,l_a^+$ satisfy constraint (\ref{eq: a3}-\ref{eq: a5}), so $w_n(\beta_a)\leq \hat{l}_a^+$ for any $a\in A$. Hence, we have $M_k(a,:)\hat{y}+\hat{\alpha}_1(k) -Q_a(k,:)\beta_a+w_n(\beta_a)\leq M_k(a,:)\hat{y}+\hat{\alpha}_1(k) -Q_a(k,:)\beta_a+\hat{l}_a^+ \leq \tilde{l}^+$ for any $k\in K$ and $a\in A$. From equation (\ref{eq: -1}-\ref{eq: -2}), we see that $\tilde{l}^+$ is a feasible solution of (\ref{eq: -1}-\ref{eq: -2}), and $\tilde{l}^+ \geq \tilde{l}^*$.

  Next, we show that $\tilde{l}^*\geq \tilde{l}^+$. It is easy to see that $(\hat{l}^*_a)_{a\in A}, (y^*_a)_{a\in A}, (l_a^*)_{a\in A}$ satisfy constraint (\ref{eq: a3}-\ref{eq: a5}). Equation (\ref{eq: 1}) implies $M^k_{a,:}\hat{y}+\hat{\alpha}_1(k) -\sum_{k'\in K}Q^a_{k,k'}\beta^{k'}_a+ \hat{l}_a^* \leq \tilde{l}^*$, for any $ a\in A,k\in K$, and hence $\tilde{l}^*, (\hat{l}^*_a)_{a\in A}$ satisfies constraint (\ref{eq: a2}). Therefore, $(\hat{l}^*_a)_{a\in A}, $ $ (y^*_a)_{a\in A},$ $ (l_a^*)_{a\in A}$ is a feasible solution of LP (\ref{eq: a1}-\ref{eq: a5}), and $\tilde{l}^* \geq \tilde{l}^+$

  Therefore, $\tilde{l}^*=\tilde{l}^+$, and the optimal values of LP (\ref{eq: 0}-\ref{eq:4}) and LP (\ref{eq: a1}-\ref{eq: a5}) are the same.

  According to Proposition \ref{lemma: recursive formula and H B independent property}, (\ref{eq: LP, regret based, dual 1}-\ref{eq: LP, regret based, dual 5}) is true.
\end{proof}

Now, we will give the detailed algorithm to compute the regret based security strategy of the uninformed player in a primal game $\Gamma_N(p)$
\begin{algorithm}\hfill{}
  \label{algorithm: regret based, primal}
\begin{enumerate}
  \item Initialization
    \begin{enumerate}
      \item Read payoff matrices $M$, transition matrices $Q$, time horizon $N$ and initial probability $p$.
      \item Set $t=1$, and $p_t=p$.
    \end{enumerate}
  \item Compute the regret $\hat{\alpha}_1=-l_1^*$ at stage 1 by solving LP (\ref{eq: LP, uninformed, primal 1}-\ref{eq: LP, uninformed, primal 3}).
  \item Compute the security strategy $\hat{y}^*$ and the regret candidate $(\beta^*_a)_{a\in A}$ by solving LP (\ref{eq: LP, regret based, dual 1}-\ref{eq: LP, regret based, dual 5}) with $n=N-t$ and $\hat{\alpha}_1=\hat{\alpha}_t$.
  \item Draw an action in $B$ according to $\hat{y}^*$, and read player 1's action $a_t$.
  \item Update $t=t+1$, and set $\hat{\alpha}_t=\beta^*_{a_t}$.
  \item If $t\leq N$, go to step 3. Otherwise, end.
\end{enumerate}
\end{algorithm}

\section{Case study: Asymmetric information stochastic intrusion detection game}
\label{section: simulation}
Reference \cite{alpcan2010network} introduced a stochastic intrusion detection game. In this game, an administrator is assigned to protect a system from attacks. The administrator can do either high level maintenance (hl) or low level maintenance (ll) at every time stage, with $A=\{hl,ll\}$. If high level maintenance is done, the system is less vulnerable to attacks. Otherwise, the system is more vulnerable to attacks. We indicate the state of the system as vulnerable (v) or non-vulnerable (nv), i.e. $K=\{nv,v\}$. The transition matrices are given in Table \ref{table: transition matrix}. The attacker decides whether to launch an attack (a) or not (na) at every stage, i.e. $B=\{a,na\}$. The corresponding payoff of a vulnerable system is always lower than the payoff of a non-vulnerable system, which is reflected by the payoff matrices in Table \ref{table: payoff matrices}. While in \cite{alpcan2010network}, it is assumed that the attacker knows the original system state, we assume that the attacker cannot directly observe the system state at any stage. Besides the transition matrices, payoff matrices and the initial probability over the state, the attacker knows the administrator's action at every stage. But the payoff, or the actual influence due to the attack, is not known by the attacker. We model this intrusion detection problem as an asymmetric information stochastic game, and demonstrate our main results about history based and belief based security strategies of informed player, and history based and regret based security strategies of uninformed player in this model.
\begin{table}
\caption{Transition matrices $Q_a$ of stochastic intrusion detection game}
\label{table: transition matrix}
\center
\begin{tabular}{ccccccc}
  &      nv &        v& & &nv &v \\
  \cline{2-3} \cline{6-7}
nv&\multicolumn{1}{|c}{0.9}&\multicolumn{1}{c|}{0.1}& & nv&\multicolumn{1}{|c}{0.8}&\multicolumn{1}{c|}{0.2} \\
v &\multicolumn{1}{|c}{0.1}&\multicolumn{1}{c|}{0.9}& &
v&\multicolumn{1}{|c}{0.2}&\multicolumn{1}{c|}{0.8} \\
  \cline{2-3} \cline{6-7}
 &\multicolumn{2}{c}{$Q_{hl}$}& & & \multicolumn{2}{c}{$Q_{ll}$}
\end{tabular}
\end{table}
\begin{table}
  \caption{Payoff matrices $M_k$ of stochastic intrusion detection game}
  \label{table: payoff matrices}
  \center
  \begin{tabular}{ccccccc}
    & a & na & & &a & na \\
    \cline{2-3}\cline{6-7}
    hl&\multicolumn{1}{|c}{3}&\multicolumn{1}{c|}{-10}& & hl&\multicolumn{1}{|c}{3}&\multicolumn{1}{c|}{-11} \\
    ll &\multicolumn{1}{|c}{-1}&\multicolumn{1}{c|}{0}& &
    ll&\multicolumn{1}{|c}{-2}&\multicolumn{1}{c|}{0} \\
      \cline{2-3} \cline{6-7}
     &\multicolumn{2}{c}{$M_{nv}$}& & & \multicolumn{2}{c}{$M_{v}$}
  \end{tabular}
\end{table}

We set the time horizon $N=3$ and the initial probability $p_0=[0.5\ 0.5]$. The history based security strategy of the administrator (informed player) computed according to Theorem \ref{theorem: history based strategy} is given in the first two rows in Table \ref{table: security strategy, informed}, and the game value is $-3.4698$. We, then, follow Algorithm \ref{Algorithm: belief based, informed} to compute the belief based security strategy which is the same as the history based security strategy, and the updated belief of the corresponding $h_t^A$ is given in the last row of Table \ref{table: security strategy, informed}. This demonstrates Theorem \ref{theorem: sufficien statistics, informed} which says that the informed player has a security strategy that only depends on the current stage and the current belief.
\begin{table}
  \caption{Security strategy of informed player and belief update. Each element of the first two rows is administrator's security strategy $\sigma_t^{*T}(k_t,I_t)$. Each element of the last row is the belief $p_t^T$ given informed player's security strategy $\sigma^*$ and history action $I_t$. }
    \label{table: security strategy, informed}
  \center
    \begin{tabular}{|c|cc|cc|cc|cc|}
    \hline
     \backslashbox{$k_t$}{$I_t$}& \multicolumn{2}{|c|}{$\emptyset$}& \multicolumn{2}{|c|}{1}& \multicolumn{2}{|c|}{2}& \multicolumn{2}{|c|}{1,1} \\    \hline
    nv & 0& 1& 0 & 1& 0&1 &0&1\\ \hline
    v & 0.1875& 0.8125& 0.375&0.625& 0.1356& 0.8644& 0.375& 0.625\\ \hline
    belief & 0.5 & 0.5 & 0.8& 0.2& 0.1448 & 0.8552 & 0.8 &0.2 \\ \hline\hline
    \backslashbox{$k_t$}{$h_t^A$}& \multicolumn{2}{|c|}{1,2} & \multicolumn{2}{|c|}{2,1}& \multicolumn{2}{|c|}{2,2}& & \\ \hline
    v&0&1&0&1&0&1 & &\\ \hline
    nv  & 0.133& 0.867 & 0.375 & 0.625& 0.1391& 0.8609 & & \\ \hline
    belief  &0.1135& 0.8865 & 0.8& 0.2& 0.1836& 0.8164& & \\ \hline
    \end{tabular}
\end{table}

The history based security strategy of the attacker (uninformed player) computed according to Theorem \ref{theorem: LP, history based, primal, uninformed} is provided in the first row in Table \ref{table: security strategy, uninformed}, and the computed game value is $-3.4698$ which meets the game value computed by the informed player. We then compute the regret based security strategy of attacker according to Algorithm \ref{algorithm: regret based, primal}. The regret based security strategy is the same as the history based security strategy, which demonstrate Corollary \ref{corollary: regret based, primal, uninformed}. The initial regret and updated regret given $h_t^A$ is provided in the last row of Table \ref{table: security strategy, uninformed}.
\begin{table}
  \caption{Security strategy of uninformed player and regret update. Each element in the first row is informed player's security strategy $\tau^{*T}_t(I_t)$, and each element in the last row is the regret $\hat{\alpha}_t$ given uninformed player (attacker)'s security strategy $\tau^*$ and informed player (administrator)'s history action $I_t$}
    \label{table: security strategy, uninformed}
  \center
    \begin{tabular}{|c|cc|cc|cc|cc|}
    \hline
     $h_t^A$& \multicolumn{2}{|c|}{$\emptyset$}& \multicolumn{2}{|c|}{1}& \multicolumn{2}{|c|}{2}& \multicolumn{2}{|c|}{1,1} \\    \hline
    $\tau_t$ & 0.6657& 0.3347& 0.6617& 0.3383& 0.6617& 0.3383& 0.6875& 0.3125\\ \hline
    regret & 3.1669& 3.7729& 0.9109& 1.5038& 0.7568& 1.3498& 0.2621& 0.9496  \\ \hline \hline
    $h_t^A$&  \multicolumn{2}{|c|}{1,2} & \multicolumn{2}{|c|}{2,1}& \multicolumn{2}{|c|}{2,2}& &\\ \hline
    $\tau_t$& 0.6875& 0.3125& 0.6875& 0.3125& 0.6875& 0.3125 & & \\ \hline
    regret &0.0666& 0.7541& 0.2621& 0.9496& 0.0666& 0.7541& & \\ \hline
    \end{tabular}
\end{table}

The security strategies of both players are then used in the intrusion detection game. We considered $1000$ realizations of the game, and compute the empirical average of the attacker's payoff over those runs. This number, $-3.4204$, is comparable to the computed value of the game, $-3.4698$, which demonstrates that the strategies of both players shown in Table \ref{table: security strategy, informed} and \ref{table: security strategy, uninformed} achieves the game value, and hence are indeed the security strategies of the corresponding players.

\section{Conclusion and future work}
\label{section: conclusion}
This paper studied security strategies of both players in an asymmetric information zero-sum two-player stochastic game in which only the informed player controls the system state's evolution. We showed that security strategies exist for the informed player, which only depend on the belief, and can be computed by solving a linear program whose size is linear in the cardinality of the uninformed player's action set. A similarly computationally attractive security strategy also exists for the uninformed player, which only depends on a new object called `the regret', which can itself be efficiently computed at every step.

We are interested in extending this work to two-player Bayesian stochastic games where both players have their own private types. The main foreseen challenge in Bayesian stochastic games is to find out sufficient statistics of both players. We expect regret to play a role in this context as well, maybe in combination with a player's belief on its own type.

\bibliography{}

\end{document}